\documentclass[conference]{IEEEtran}

\usepackage{amsmath,amssymb,amscd}
\usepackage{subfigure, epsfig}
\usepackage{graphicx}
\usepackage{pstricks}
\usepackage{pstricks-add}
\usepackage{pst-plot}

\newcommand{\mc}{\mathcal}

\newtheorem{theorem}{Theorem}[section]

\newtheorem{definition}[theorem]{Definition}

\newtheorem{corollary}[theorem]{Corollary}

\begin{document}

\title{``To sense" or ``not to sense" in energy-efficient power control games}
%
%
%
%
%

%

\author{\IEEEauthorblockN{Ma\"{e}l Le Treust\IEEEauthorrefmark{1},
Yezekael Hayel\IEEEauthorrefmark{2},
Samson Lasaulce\IEEEauthorrefmark{1} and
M\'{e}rouane Debbah \IEEEauthorrefmark{3}}
\IEEEauthorblockA{\IEEEauthorrefmark{1}Laboratoire des Signaux et Syst\`{e}mes,
CNRS  - Universit\'{e} Paris-Sud 11 - Sup\'{e}lec,
91191, Gif-sur-Yvette Cedex, France\\
Email: \{letreust\},\{lasaulce\}@lss.supelec.fr}
\IEEEauthorblockA{\IEEEauthorrefmark{2}
Laboratoire d'information d'Avignon,
Universit\'{e} d'Avignon,
84911 Avignon, France\\
Email: yezekael.hayel@univ-avignon.fr}
\IEEEauthorblockA{\IEEEauthorrefmark{3}
Chaire Alcatel-Lucent,
SUPELEC,
91190 Gif-sur-Yvette, France\\
Email: debbah@supelec.fr}}



\maketitle

\begin{abstract}
A network of cognitive transmitters is considered. Each transmitter
has to decide his power control policy in order to maximize
energy-efficiency of his transmission. For this, a transmitter has
two actions to take. He has to decide whether to sense the power
levels of the others or not (which corresponds to a finite sensing
game), and to choose his transmit power level for each block (which
corresponds to a compact power control game). The sensing game is
shown to be a weighted potential game and its set of correlated
equilibria is studied. Interestingly, it is shown that the general
hybrid game where each transmitter can jointly choose the hybrid
pair of actions (to sense or not to sense, transmit power level)
leads to an outcome which is worse than the one obtained by playing
the sensing game first, and then playing the power control game.
This is an interesting Braess-type paradox to be aware of for
energy-efficient power control in cognitive networks.
\end{abstract}


\section{Introduction}

In fixed communication networks, the paradigm of peer-to-peer
communications has known a powerful surge of interest during the the
two past decades with applications such as the Internet. Remarkably,
this paradigm has also been found to be very useful for wireless
networks. Wireless ad hoc and sensor networks are two illustrative
examples of this. One important typical feature of these networks is
that the terminals have to take some decisions in an autonomous
(quasi-autonomous) manner. Typically, they have to choose their
power control and resources allocation policy. The corresponding
framework, which is the one of this paper, is the one of distributed
power control or resources allocation. More specifically, the
scenario of interest is the case of power control in cognitive
networks. Transmitters are assumed to be able to sense the power
levels of neighboring transmitters and adapt their power level
accordingly. The performance metric for a transmitter is the
energy-efficiency of the transmission \cite{goodman-pc-2000} that is, the
number of bits successfully decoded by the receiver per Joule
consumed at the transmitter.

The model of multiuser networks which is considered is a multiple
access channel with time-selective non-frequency selective links.
Therefore, the focus is not on the problem of resources allocation
but only on the problem of controlling the transmit power over
quasi-static channels. The approach of the paper is related to the
one of \cite{lasaulce-twc-2009}\cite{HeLasaulceHayel-infocom11} where some hierarchy
is present in the network in the sense that some transmitters can
observe the others or not; also the problem is modeled by a game
where the players are the transmitters and the strategies are the
power control policies. One the differences with
\cite{lasaulce-twc-2009}\cite{HeLasaulceHayel-infocom11} is that \emph{every}
transmitter can be cognitive and sense the others but
observing/sensing the others has a cost. Additionally, a new type of
power control games is introduced (called hybrid power control
games) in which an action for a player has a discrete component
namely, to sense or not to sense, and a compact component namely,
the transmit power level. There are no general results for
equilibrium analysis in the game-theoretic literature. This is a
reason why some results are given in the 2-player case only, as a
starting point for other studies. In particular, it is shown that it
is more beneficial for every transmitter to choose his discrete
action first and then his power level. The (finite) sensing game is
therefore introduced here for the first time and an equilibrium
analysis is conducted for it. Correlated equilibria are considered
because they allow the network designer to play with fairness, which
is not possible with pure or mixed Nash equilibria.

This paper is structured as follows. A review of the previous results
regarding the one-shot energy efficient power control game is presented
in Sec. 2. The sensing game is formally defined and some equilibrium
 results are stated in Sec. 3. A detailed analysis of the
 2-players sensing is provided in Sec. 4 and the conclusion appears in Sec. 5.

\section{Review of known results}

\subsection{Review of the one-shot energy-efficient power control game (without
sensing)}

We review a few key results from \cite{GoodMandPCWD00}
concerning the static non-cooperative PC game. In order to define
the static PC game some notations need to be introduced. We denote
by $R_i$ the transmission information rate (in bps) for user $i$ and
$f$ an efficiency function representing the block success rate,
which is assumed to be sigmoidal and identical for all the users;
the sigmoidness assumption is a reasonable assumption, which is well
justified in
\cite{Rodriguez-globecom-2003}\cite{Meshkati-tcom-2005}. Recently,
\cite{Belmega-valuetools-2009} has shown that this assumption is
also justified from an information-theoretic standpoint. At a given
instant,
 the SINR at receiver $i \in \mc{K}$ writes as:
\begin{equation}
\label{eq:sinr-ne} \mathrm{SINR}_i=\frac{p_i|g_i|^2}{\sum_{j \neq i}
p_j|g_j|^2+\sigma^2}
\end{equation}
where $p_i$ is the power level for transmitter $i$, $g_i$ the
channel gain of the link between transmitter $i$ and the receiver,
$\sigma^2$ the noise level at the receiver, and $f$ is a sigmodial
efficiency function corresponding to the block success rate. With
these notations, the static PC game, called $\mc{G}$, is defined in
its normal form as follows.

\begin{definition}[Static PC game] \emph{The static PC game is a triplet
$\mc{G} = (\mc{K}, \{\mc{A}_i\}_{i\in\mc{K}},\{u_i\}_{i\in\mc{K}})$
where $\mc{K}$ is the set of players, $\mc{A}_1,...,\mc{A}_K$ are
the corresponding sets of actions, $\mc{A}_i = [0,
P_i^{\mathrm{max}}]$, $P_i^{\mathrm{max}}$ is the maximum transmit
power for player $i$, and $u_1,...,u_k$ are the utilities of the
different players which are defined by:}
\begin{eqnarray}
\label{eq:def-of-utility} u_i(p_1,...,p_K)= \frac{R_i
f(\mathrm{SINR}_i)}{p_i} \ [\mathrm{bit} / \mathrm{J}].
\end{eqnarray}
\end{definition}
In this game with complete information ($\mc{G}$ is known to every
player) and rational players (every player does the best for himself
and knows the others do so and so on), an important game solution
concept is the NE (i.e., a point from which no player has interest
in unilaterally deviating). When it exists, the non-saturated NE of
this game can by obtained by setting $\frac{\partial u_i}{\partial
p_i}$ to zero, which gives an equivalent condition on the SINR: the
best SINR in terms of energy-efficiency for transmitter $i$ has to
be a solution of $xf'(x)-f(x)=0$ (this solution is independent of
the player index since a common efficiency function is assumed, see
\cite{Meshkati-tcom-2005} for more details). This leads to:
\begin{equation} \forall i \in \{1,...,K \}, \
p_i^{*}= \frac{\sigma^2}{|g_i|^2} \frac{\beta^*}{1-(K-1)\beta^{*}}
\label{eq:NE-power}
\end{equation}
where $\beta^{*}$ is the unique solution of the equation
$xf'(x)-f(x)=0$. By using the term ``non-saturated NE'' we mean that
the maximum transmit power for each user, denoted by $
P_i^{\mathrm{max}}$, is assumed to be sufficiently high not to be
reached at the equilibrium i.e., each user maximizes his
energy-efficiency for a value less than $P_i^{\mathrm{max}}$ (see
\cite{lasaulce-twc-2009} for more details). An important property of
the NE given by (\ref{eq:NE-power}) is that transmitters only need
to know their individual channel gain $|g_i|$ to play their
equilibrium strategy. One of the interesting results of this paper
is that it is possible to obtain a more efficient equilibrium point
by repeating the game $\mc{G}$ while keeping this key property.

\subsection{Review of the Stackelberg energy-efficient power control game (with
sensing)}

Here we review a few key results from \cite{HeLasaulceHayel-infocom11}.
The framework addressed in \cite{HeLasaulceHayel-infocom11} is that the
existence of two classes of transmitters are considered: those who
can sense and observe the others and those who cannot observe. This
establishes a certain hierarchy between the transmitters in terms of
observation. A suited model to study this is the Stackelberg game
model \cite{stackelberg-book-1934}: some players choose their transmit
power level (these are the leaders of the power control game) and
the others observe the played action and react accordingly (these
are the followers of the game). Note that the leaders know they are
observed and take this into account for deciding. This leads to a
game outcome (namely a Stackelberg equilibrium) which
Pareto-dominates the one-shot game Nash equilibrium (given by
(\ref{eq:NE-power})) when there is no cost for sensing
\cite{lasaulce-twc-2009}. However, when the fraction of time to sense
is taken to be $\alpha >0$, the data rate is weighted by
$(1-\alpha)$ and it is not always beneficial for a transmitter to
sense \cite{HeLasaulceHayel-infocom11}. The equilibrium action and utility
for player $i$ when he is a game leader (L) are respectively given
by
\begin{equation}
p_i^L=\frac{\sigma^2}{|g_i|^2}
\frac{\gamma^*(1+\beta^*)}{1-(K-1)\gamma^*\beta^*-(K-2)\beta^*}
\end{equation}
where $\gamma^*$ is the unique solution of $x
\left[1-\frac{(K-1)\beta^*}{1-(K-2)\beta^*} x\right] f'(x) - f(x)
=0$ and
\begin{equation}
u^L_i=\frac{|g_i|^2}{\sigma^2}\frac{1-(K-1)\gamma^*\beta^*
-(K-2)\beta^*}{\gamma^*(1+\beta^*)}f(\gamma^*).
\end{equation}
On the other hand, if player $i$ is a follower (F) we have that:
\begin{equation}
p_i^F=\frac{\sigma^2}{|g_i|^2}\frac{\beta^*(1
+\gamma^*)}{1-(K-1)\gamma^*\beta^*-(K-2)\beta^*}
\end{equation}
and
\begin{equation}
u^F_i= (1-\alpha)
\frac{|g_i|^2}{\sigma^2}\frac{1-(K-1)\gamma^*\beta^*
-(K-2)\beta^*}{\beta^*(1+\gamma^*)}f(\beta^*).
\end{equation}

\section{A new game: the $K-$player sensing game}

\subsection{Sensing game description}

In the two hierarchical power control described above, the
transmitter is, by construction, either a cognitive transmitter or a
non-cognitive one and the action of a player consists in choosing a
power level. Here, we consider that all transmitters can sense, the
power level to be the one at the Stackelberg equilibrium, and the
action for a player consists in choosing to sense (S) or not to
sense (NS). This game is well defined only if at least one player is
a follower (i.e., he senses) and one other is the leader (i.e., he
does not sense). We assume in the following that the total number of
transmitters is $K+2$, where $K$ transmitters are considered as
usual players and the two last are a follower and a leader. Define
the $K-$player sensing game as a triplet:
\begin{eqnarray}
G=(K,(\mathcal{S})_{i\in K},(U_i)_{i\in K})
\end{eqnarray}
where the actions set are the same for each player $i\in K$, sense
or not sense: $\mathcal{S}=(S,NS)$. The utility function of each
player $i\in K$ depends on his own channel state $g_i$ and
transmission rate $R_i$ but also on the total number of players $F$
playing the sensing action and the number of players that non sense
denoted $L$. Denote $U_i^S(F,L)$ the utility of player $i$ when
playing action sensing $S$ whereas $F-1$ other players are also
sensing and $L$ other players are non-sensing. The total number of
player is $F + L = K$.

\begin{eqnarray*}
U_i^S(F,L)&=&\frac{g_iR_i}{\sigma^2}\frac{f(\beta^*)}{N\beta^{\star}(N+\gamma_{L+1}^{\star})}\\
&&\left(N^2\!-\!N\beta^{\star} - \left[(N\!+\!\beta^{\star})L\!+ (F+1) \beta^{\star}\right]\gamma^{\star}_{L+1}\right)\\
U_i^{NS}(F,L)&=&\frac{g_iR_i}{\sigma^2}\frac{f(\gamma^*_L)}{N\gamma_{L+1}^{\star} (N + \beta^{\star})}\\
&&\left(N^2\!-\!N\beta^{\star} - \left[(N\!+\!\beta^{\star})L\!+
(F+1) \beta^{\star}\right]\gamma^{\star}_{L+1}\right)
\end{eqnarray*}
with $\gamma^*_L$ solution of $x(1-\epsilon_L x)f'(x)=f(x)$ with:
\begin{eqnarray}
\epsilon_L=\frac{(K +2 -L)\beta^{\star}}{N^2-N(K + 1
-L)\beta^{\star}}.
\end{eqnarray}

\subsection{The sensing game is a weighted potential game}

The purpose of this section is to show that the sensing game may be
an exact potential game. However, this holds under restrictive
assumptions on the channel gains. It is then shown, as a second
step, that the game is a weighted potential game. For making this
paper sufficiently self-containing we review important definitions
to know on potential games.

\begin{definition}[Monderer and Shapley 1996 \cite{Monderer-Shapley-1996}]
The normal form game $G$ is a potential game ) if there is a
potential function $V : S \longrightarrow \mathbb{R}$ such that
\begin{eqnarray}
U_i(s_i,s_{-i}) -  U_i(t_i,s_{-i}) = V(s_i,s_{-i}) -  V(t_i,s_{-i}),\\
\quad \forall i\in K,\; s_i,t_i\in \mathcal{S_i}
\end{eqnarray}
\end{definition}

\begin{theorem}
The sensing game $G=(K,(\mathcal{S})_{i\in K},(U_i)_{i\in K})$ is an
exact potential game if and only if one of the two following
conditions is satisfied.
\begin{eqnarray*}
1)&&\forall i,j\in K\quad R_ig_i=R_jg_j\\
2)&&\forall i,j\in K,\; s_i,t_i\in S_i,\;\forall s_j,t_j\in S_j,\;\forall s_k\in S_{K\backslash \{i,j\}}\\
&&U^{T}(t_i,s_j,s_k)- U^S(s_i,s_j,s_k)\\
&& + U^S(s_i,t_j,s_k) -U^T(t_i,t_j,s_k)= 0
\end{eqnarray*}
\end{theorem}

The Proof is given in the Appendix 4.

The potential functions of our game depends on which condition is
satisfied in the above theorem. Suppose that the  first condition is
satisfied $\forall i,j\in K\quad R_ig_i=R_jg_j$. Then the
Rosenthal's potential function writes :
\begin{eqnarray*}
\Phi(F,L)&=& \sum_{i=1}^F U^S(i,K-i) + \sum_{j=1}^L U^{NS}(K-j,j)
\end{eqnarray*}

\begin{theorem}[Potential Game \cite{Monderer-Shapley-1996}]
Every finite potential game is isomorphic to a congestion game.
\end{theorem}

\begin{definition}[Monderer and Shapley 1996 \cite{Monderer-Shapley-1996}]
The normal form game $G$ is a weighted potential game if there is a
vector $(w_i)_{i\in K}$ and a potential function $V : S
\longrightarrow \mathbb{R}$ such that:
\begin{eqnarray*}
&&U_i(s_i,s_{-i}) -  U_i(t_i,s_{-i}) = w_i(V(s_i,s_{-i}) -  V(t_i,s_{-i})),\\
&&\quad \forall i\in K,\; s_i,t_i\in \mathcal{S}_i
\end{eqnarray*}
\end{definition}

\begin{theorem}
The sensing game $G=(K,(\mathcal{S}_i)_{i\in K},(U_i)_{i\in K})$ is
a weighted potential game with the weight vector:
\begin{eqnarray}
\forall i\in K\quad w_i = \frac{R_ig_i}{\sigma^2}
\end{eqnarray}
\end{theorem}

The Proof is given in the Appendix 5.

\subsection{Equilibrium analysis}

First of all, note that since the game is finite (i.e., both the
number of players and the sets of actions are finite), the existence
of at least one mixed Nash equilibrium is guaranteed
\cite{nash-academy-50}. Now, since we know that the game is weighted
potential we know that there is at least one pure Nash equilibrium
\cite{Monderer-Shapley-1996}. Indeed, the following theorem holds.

\begin{theorem}
The equilibria of the above potential game is the set of maximizers
of the Rosenthal potential function \cite{Rosenthal73}.
\begin{eqnarray*}
&&\{S = (S_1,\ldots,S_K) | S\in NE\}
= \arg\max_{(F,L)} \Phi(F,L)\\
 &=& \arg\max_{(F,L)}\left[ \sum_{i=1}^F U(S,i,K-i) + \sum_{j=1}^L U(NS,K-j,j) \right]
\end{eqnarray*}
\end{theorem}
The proof follows directly the one of Rosenthal's theorem \cite{Rosenthal73}.

We may restrict our attention to pure and mixed Nash equilibria.
However, as it will be clearly seen in the 2-player case study (Sec.
\ref{subsec:2playerSensing}), this may pose a problem of fairness. This is the main reason
why we study the set of correlated equilibria of the sensing game.
We introduce the concept of correlated equilibrium \cite{Aum74}
in order to enlarge the set of equilibrium utilities.
Every utility vector inside the convex hull of the
equilibrium utilities is a correlated equilibrium. The
convexification property of the correlated equilibrium allow the
system to better chose an optimal sensing. The concept of correlated
equilibrium is a generalization of the Nash equilibrium. It consist
in the stage game $G$ extended with a signalling structure $\Gamma$.
A correlated equilibrium (CE) of a stage game correspond to a Nash
equilibrium (NE) of the same game extended with an adequate
signalling structure $\Gamma$. A canonical correlated equilibrium is
a probability distribution $Q\in \Delta(A)$, $A=A_1\times...\times
A_K$ over the action product of the players that satisfy some
incentives conditions.

\begin{definition}
A probability distribution $Q\in \Delta(A)$ is a canonical
correlated equilibrium if for each player $i$, for each action
$a_i\in A_i$ that satisfies $Q(a_i)>0$ we have:
\begin{eqnarray*}\label{CorrelatedLinearInequlity}
\sum_{a_{-i} \in A_{-i}}&& Q(a_{-i} \mid a_i)u_i(a_i,a_{-i})\\
&&\geq \sum_{a{-i} \in A_{-i}}Q(a_{-i} \mid a_i)u_i(b_i,a_{-i}),\\
&&\quad \forall b_i\in A_i
\end{eqnarray*}
\end{definition}
The result of Aumann 1987 \cite{Aum87} states that for any correlated
equilibrium, it correspond a canonical correlated equilibrium.
\begin{theorem}[Aumann 1987, prop. 2.3 \cite{Aum87}]\label{theorem:correlated}
The utility vector $u$ is a correlated equilibrium utility if and
only if there exists a distribution $Q\in \Delta(A)$ satisfying the
linear inequality contraint \ref{CorrelatedLinearInequlity} with $u=
E_Q U$.
\end{theorem}

The convexification property of the correlated equilibrium allow the
system to better chose an optimal sensing. Denote $E$ the set of
pure or mixed equilibrium utility vectors and $\text{Conv }E$ the
convex hull of the set $E$.

\begin{theorem}
Every utility vector $u\in  \text{Conv }E$ is a correlated
equilibrium utility of the sensing game.
\end{theorem}

Any convex combination of Nash equilibria is a correlated
equilibrium. As example, let $(\underline{U}_j)_{j\in J}$ a family
of equilibrium utilities and $(\lambda_j)_{j\in J}$ a family of
positive parameters with  $\sum_{j\in J}\lambda_j=1$ such that:
\begin{eqnarray}
\underline{U} &=&  \sum_{j\in J}\lambda_j \underline{U}_j
\end{eqnarray}
Then $\underline{U}$ is a correlated equilibrium utility vector.

\section{Detailed analysis for the 2-player case}

\subsection{The 2-player hybrid power control game}


In the previous section, we consider the sensing game as if the
players do not chose their own power control policy. Indeed, when a
player chooses to sense, he cannot choose its own power control
because, it would depend on whether the other transmitters sense or
not.  We investigate the case where the players are choosing their
sensing and power control policy in a joint manner. It enlarges the
set of actions of the sensing game and it turns that, as a
Braess-type paradox, that the set of equilibria is dramatically
reduced. The sensing game with power control has a stricly dominated
strategy: the sensing strategy. It implies that the equilibria of
such a game boils down to the Nash equilibrium without sensing.

We consider that the action set for player $i$ consists in choosing
to sense or not and the transmit power level. The action set of
player $i$ writes :
\begin{eqnarray}
A_i=  \{S_i,NS_i\}\times [0,\bar{P_i}]
\end{eqnarray}

Before to characterize the set of equilibria of such a game, remark
that the two pure equilibria of the previous matrix game are no
longer equilibria. Indeed, assume that player 2 will not sense its
environment and transmit  using the leading power $p_2^L$. Then
player 1 best response would be to play the following transmit power
$p_1^F$ as for the classical Stackelberg equilibrium. Nevertheless
in the above formulation, the player 1 has a sensing cost $\alpha$
that correspond to the fraction of time to sense its environment. In
this context, player 1 is incited to play the following transition
power without sensing. The strategy $(S_1,p_1^F)$ and $(NS_2,p_1^L)$
is not an equilibrium of the game with Discrete and Compact Action
Set.

\begin{theorem}
The unique Nash equilibrium of the Power Control and Sensing Game is
the Nash equilibrium without sensing.
\end{theorem}

\begin{proof}
This result comes from the cost of sensing activity. Indeed, the
strategy $(S_1,p_1)$ is always dominated by the strategy
$(NS_1,p_1)$. It turns out that the sensing is a dominated actions
for both players 1 and 2. Thus every equilibria is of the form
$(NS_1,p_1)$, $(NS_2,p_2)$ with the reduced action spaces $p_1\in
[0,\bar{P_1}]$ and $p_2\in [0,\bar{P_2}]$. The previous analysis
applies in that case, showing that the unique Nash equilibrium of
the Power Control and Sensing Game is the Nash of the game without
sensing $(p_1^*,p_2^*)$.
\end{proof}

As a conclusion, we see that letting the choice to the transmitters
to choose jointly their discrete and continuous actions lead to a
performance which is less than the one obtained by choosing his
discrete action first, and then choosing his continuous action. This
is the reason why we assume, from now on, the existence of a
mechanism imposing this order in the decision taking.

\subsection{The 2-player sensing game}\label{subsec:2playerSensing}

We consider the following two players-two strategies matrix game where
players 1 and 2 choose to sense the channel (action $S$) or not
(action $NS$) before transmitting his data. We denote by $x_i$ the
 mixed strategy of user $i$, that is the probability that user $i$
 takes action $S$ (sense the channel). Sensing activity provide the
  possibility to play as a follower, knowing in advance the action
  of the leaders.  Let $\alpha$ denote the sensing cost, we compare
  the strategic behavior of sensing by considering the equilibrium
   utilities at the Nash and at the Stackelberg equilibria as payoff functions.

\begin{figure}[!ht]
\begin{center}
\begin{tiny}
\psset{xunit=0.5cm,yunit=0.2cm}
\begin{pspicture}(0,-1)(19,19)
\psframe(0,0)(16,16)
\psline(0,8)(16,8)
\psline(8,0)(8,16)
\rput(3,13.3){$\frac{R_1g_1f(\beta^*)(1-\beta^*)}{\sigma^2\beta^*}, $}
\rput(5.5,10.6){$\frac{R_2g_2f(\beta^*)(1-\beta^*)}{\sigma^2\beta^*}$}
\rput(10.6,13.3){$\frac{R_1g_1f(\gamma^*)(1-\gamma^*\beta^*)}{\sigma^2\gamma^*(1+\beta^*)}, $}
 \rput(12.6,10.3){$(1-\alpha)\frac{R_2g_2f(\beta^*)(1-\gamma^*\beta^*)}{\sigma^2\beta^*(1+\gamma^*)}$}
\rput(3.5,5.8){$(1-\alpha)\frac{R_1g_1f(\beta^*)(1-\gamma^*\beta^*)}{\sigma^2\beta^*(1+\gamma^*)},  $}
 \rput(5.5,3){$\frac{R_2g_2f(\gamma^*)(1-\gamma^*\beta^*)}{\sigma^2\gamma^*(1+\beta^*)}$}
\rput(11,5.8){$(1-\alpha)\frac{R_1g_1f(\beta^*)(1-\beta^*)}{\sigma^2\beta^*}, $}
\rput(13,3){$(1-\alpha)\frac{R_2g_2f(\beta^*)(1-\beta^*)}{\sigma^2\beta^*}$}
\rput(-0.6,4){$S_1$}
\rput(-0.6,12){$NS_1$}
\rput(4,17){$NS_2$}
\rput(12,17){$S_2$}
\end{pspicture}
\end{tiny}
\caption{The Utility Matrix of the Two-Player Sensing Game.}
\end{center}
\label{figure:utilMatrix}
\end{figure}

The equilibria of this game are strongly related to the sensing parameter $\alpha$.

\begin{theorem}
The matrix game has three equilibria if and only if
\begin{eqnarray}
\alpha<\frac{\beta^* - \gamma^*}{1-\beta^*\gamma^*}
\end{eqnarray}
\end{theorem}

Let us characterize the three equilibria.
From Appendix 1, is it easy to see that :
\begin{eqnarray*}
&&\alpha<\frac{\beta^* - \gamma^*}{1-\beta^*\gamma^*} \Longleftrightarrow \\ &&(1-\alpha)\frac{R_1g_1f(\beta^*)(1-\gamma^*\beta^*)}{\sigma^2\beta^*(1+\gamma^*)}>\frac{R_1g_1f(\beta^*)(1-\beta^*)}{\sigma^2\beta^*}
\end{eqnarray*}
We conclude that the joint actions $(NS_1,NS_2)$ and $(S_1,S_2)$ are not Nash Equilibria:
\begin{eqnarray}
U_1 (NS_1,NS_2) < U_1 (S_1,NS_2)\\
U_2 (NS_1,NS_2) < U_2 (NS_1,S_2)\\
U_1 (S_1,S_2) < U_1 (NS_1,S_2)\\
U_2 (S_1,S_2) < U_2 (S_1,NS_2)
\end{eqnarray}

The sensing parameter determines which one of the two options is optimal between leading and following.
\begin{corollary}
Following is better than leading if and only if
\begin{eqnarray}
\alpha<\frac{f(\beta^*) - f(\gamma^*) + \frac{f(\beta^*)}{\beta^*} - \frac{f(\gamma^*)}{\gamma^*}  }{f(\beta^*)\frac{1+\beta^*}{\beta^*}}
\end{eqnarray}
\end{corollary}
The proof is given in Appendix 3.

The above matrix game has two pure equilibria $(NS_1,S_2)$ and $(S_1,NS_2)$. There is also a completely mixed equilibrium we compute using the indifference principle.
Let $(x,1-x)$ a mixed strategy of player 1 and $(y,1-y)$ a mixed strategy of player 2. We aim at characterize the optimal joint mixed strategy $(x^*,y^*)$ satisfying the indifference principle (see Appendix 2 for more details).
\begin{figure*}[!ht]
\begin{eqnarray*}
x^*=y^* = \frac{(1-\alpha)\frac{f(\beta^*)}{\beta^*}(1-\beta^*) - \frac{f(\gamma^*)}{\gamma^*} \frac{1- \gamma^*\beta^*}{ 1+ \beta^*}  }{(1-\alpha)\frac{f(\beta^*)}{\beta^*}(1-\beta^*) - \frac{f(\gamma^*)}{\gamma^*} \frac{1- \gamma^*\beta^*}{ 1+ \beta^*}   + \frac{f(\beta^*)}{\beta^*}(1-\beta^*) - (1-\alpha) \frac{f(\beta^*)}{\beta^*} \frac{1- \gamma^*\beta^*}{ 1+ \gamma^*} }
\end{eqnarray*}
\end{figure*}
The above joint mixed strategy $(x^*,1-x^*)$ and $(y^*,1-y^*)$ is an equilibrium strategy. The corresponding utilities are computed in Appendix 2. and writes with $\Delta$ defined in(\ref{Equation:Delta}).
\begin{eqnarray*}
U_1(x^*,y^*) &=& \frac{R_1g_1}{\sigma^2} \Delta\\
U_2(x^*,y^*) &=& \frac{R_2g_2}{\sigma^2} \Delta
\end{eqnarray*}
\begin{figure*}[!ht]\label{Equation:Delta}
\begin{eqnarray*}
\Delta = \frac{(1-\alpha)\frac{f(\beta^*)}{\beta^*}(1-\beta^*)\frac{f(\beta^*)}{\beta^*}(1-\beta^*) - \frac{f(\gamma^*)}{\gamma^*} \frac{1- \gamma^*\beta^*}{ 1+ \beta^*} (1-\alpha) \frac{f(\beta^*)}{\beta^*} \frac{1- \gamma^*\beta^*}{ 1+ \gamma^*}
}{(1-\alpha)\frac{f(\beta^*)}{\beta^*}(1-\beta^*) - \frac{f(\gamma^*)}{\gamma^*} \frac{1- \gamma^*\beta^*}{ 1+ \beta^*}   + \frac{f(\beta^*)}{\beta^*}(1-\beta^*) - (1-\alpha) \frac{f(\beta^*)}{\beta^*} \frac{1- \gamma^*\beta^*}{ 1+ \gamma^*} }\\
\end{eqnarray*}
\end{figure*}


The equilibrium utilities are represented on the following figure. The two pure Nash equilibrium utilities are represented by a circle whereas the mixed Nash utility is represented by a square.
\begin{figure}[ht]
\begin{center}
\psset{xunit=1.5cm,yunit=1.5cm}
\begin{pspicture}(-1,-1)(5,5)
\pscustom[fillstyle=solid,fillcolor=gray]{
\psline(0.2,0.2)(4,2)
\pscurve(4,2)(1.7,1.7)(2,4)
}
\psline[arrows=->](0,-1)(0,5)
\psline[arrows=->](-1,0)(5,0)
\psdots(4,2)(2,4)(1.7,1.7)(0.2,0.2)(0.5,0.5)
\pscircle(2,4){0.1}
\pscircle(4,2){0.1}
\psframe(1.65,1.65)(1.75,1.75)
\psline(0.2,0.2)(4,2)
\psline(0.2,0.2)(2,4)
\pscurve(4,2)(1.7,1.7)(2,4)
\uput[l](0,1.7){$U_2(x^*,y^*)$}
\uput[l](0,2){$U_2(S_1,NS_2)$}
\uput[l](0,4){$U_2(NS_1,S_2)$}
\uput[l](0,0.2){$U_2(S_1,S_2)$}
\uput[l](0,0.5){$U_2(NS_1,NS_2)$}
\uput[d](1.7,0){$U_1(x^*,y^*)$}
\uput[d](2.1,-0.3){$U_1(S_1,NS_2)$}
\uput[d](4,0){$U_1(NS_1,S_2)$}
\uput[d](0.2,0){$U_1(S_1,S_2)$}
\uput[d](0.7,-0.3){$U_1(NS_1,NS_2)$}
\psline[linestyle=dashed,linewidth=0.5pt](2,0)(2,4)
\psline[linestyle=dashed,linewidth=0.5pt](0,4)(2,4)
\psline[linestyle=dashed,linewidth=0.5pt](4,0)(4,2)
\psline[linestyle=dashed,linewidth=0.5pt](0,2)(4,2)
\psline[linestyle=dashed,linewidth=0.5pt](1.7,0)(1.7,1.7)
\psline[linestyle=dashed,linewidth=0.5pt](0,1.7)(1.7,1.7)
\psline[linestyle=dashed,linewidth=0.5pt](0.2,0)(0.2,0.2)
\psline[linestyle=dashed,linewidth=0.5pt](0,0.2)(0.2,0.2)
\psline[linestyle=dashed,linewidth=0.5pt](0.5,0)(0.5,0.5)
\psline[linestyle=dashed,linewidth=0.5pt](0,0.5)(0.5,0.5)
\label{fig:EqandFeasibleUtilities}
\end{pspicture}
\caption{The Equilibrium and Feasible Utilities.}
\end{center}
\end{figure}
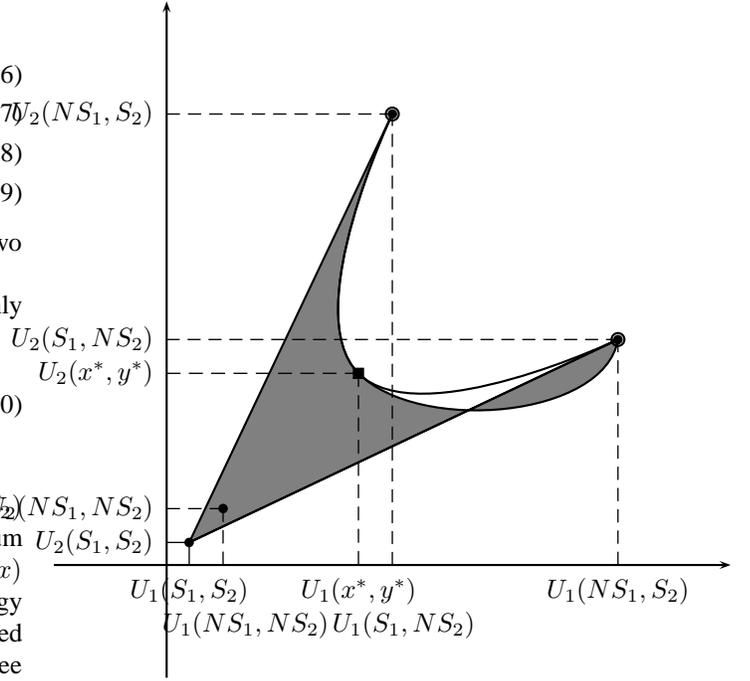

We also provide a characterization of the equilibria for the cases where $\alpha$
is greater or equal than $\frac{\beta^* - \gamma^*}{1-\beta^*\gamma^*}$.
\begin{corollary}
The matrix game has a unique equilibrium if and only if
\begin{eqnarray}
\alpha>\frac{\beta^* - \gamma^*}{1-\beta^*\gamma^*}
\end{eqnarray}
It has a infinity of equilibria if and only if
\begin{eqnarray}
\alpha=\frac{\beta^* - \gamma^*}{1-\beta^*\gamma^*}
\end{eqnarray}
\end{corollary}

First note that if the sensing cost is too high, the gain in
 terms of utility at Stackelberg instead of Nash equilibrium would be
 dominated by the loss of utility due to the sensing activity. In that case,
 the Nash equilibrium would be more efficient. Second remark that in
 case of equality, the action profiles $(NS_1,NS_2)$, $(NS_1,S_2)$, $(S_1,NS_2)$
 and every convex combination of the corresponding payoffs are all equilibrium payoffs.

Now that we have fully characterized the pure and mixed equilibria
of the game, let us turn our attention to correlated equilibria.

Theorem (\ref{theorem:correlated}) allows us to characterize the correlated equilibrium
 utility using the system of linear inequalities (\ref{CorrelatedLinearInequlity}). We
 investigate the situation where the stage game has three Nash equilibria and following
 is better than leading. We suppose that the parameter $\alpha$ satisfies.
\begin{eqnarray}
\alpha<\min(\frac{\beta^* - \gamma^*}{1-\beta^*\gamma^*}, \frac{f(\beta^*) - f(\gamma^*)
+ \frac{f(\beta^*)}{\beta^*} - \frac{f(\gamma^*)}{\gamma^*}  }{f(\beta^*)\frac{1+\beta^*}{\beta^*}} )
\end{eqnarray}
Note that the analysis is similar in the case where Leading is better than Following. However, if the parameter $\alpha>\frac{\beta^* - \gamma^*}{1-\beta^*\gamma^*}$ we have seen that the stage game has only one Nash equilibrium corresponding to play the Nash equilibrium power in the one-shot game. In such a case, no signalling device can increase the set of equilibria. The unique correlated equilibrium is the Nash equilibrium. We characterize an infinity of correlated equilibria.

\begin{theorem}
Any convex combination of Nash equilibria is a correlated equilibrium. In particular if there exists a utility vector $u$  and a parameter $\lambda\in [0,1]$ such that:
\begin{eqnarray}
u_1 &=&  \lambda U_1(S_1,NS_2) + (1 - \lambda ) U_1(NS_1,S_2)\\
u_2 &=&  \lambda U_2(S_1,NS_2) + (1 - \lambda ) U_2(NS_1,S_2)
\end{eqnarray}
Then $u$ is a correlated equilibrium.
\end{theorem}

The above result state that any distribution $Q$ defined as follows with $\lambda\in [0,1]$ is a correlated equilibrium.
\begin{figure}[!ht]
\begin{center}
\psset{xunit=0.15cm,yunit=0.1cm}
\begin{pspicture}(-1,-1)(19,19)
\psframe(0,0)(16,16)
\psline(0,8)(16,8)
\psline(8,0)(8,16)
\rput(4,12){$0$}
\rput(12,12){$1-\lambda$}
\rput(4,4){$\lambda$}
\rput(12,4){$0$}
\rput(-1.7,5){$S_1$}
\rput(-2.2,12){$NS_1$}
\rput(4,18){$NS_2$}
\rput(12,18){$S_2$}
\end{pspicture}
\end{center}
\end{figure}
The canonical signalling device which should be added to the game consist in a lottery
with parameter $\lambda$ over the actions $(S_1,NS_2)$ and $(NS_1,S_2)$ and of signalling
structure such that each player receives her component. For example, if $(S_1,NS_2)$
is chosen the player 1 receives the signal ``play $S_1$" whereas player 2 receives the signal ``play $NS_2$".

The correlated equilibrium utilities are represented by the bold line. The signalling
device increase the achievable utility region by adding the light gray area.
\begin{footnotesize}
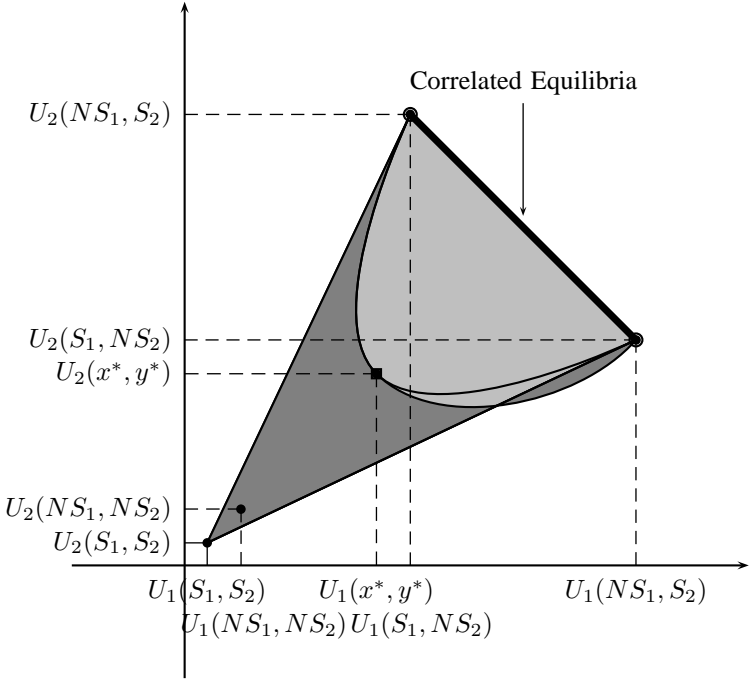
\begin{figure}[!ht]
\begin{center}
\psset{xunit=1.5cm,yunit=1.5cm}
\begin{pspicture}(-1,-1)(5,5)
\pspolygon[fillstyle=solid,fillcolor=lightgray]  (0.2,0.2)(4,2)(2,4)
\pscustom[fillstyle=solid,fillcolor=gray]{
\psline(0.2,0.2)(4,2)
\pscurve(4,2)(1.7,1.7)(2,4)
}
\psline[arrows=->](0,-1)(0,5)
\psline[arrows=->](-1,0)(5,0)
\psdots(4,2)(2,4)(1.7,1.7)(0.2,0.2)(0.5,0.5)
\pscircle(2,4){0.1}
\pscircle(4,2){0.1}
\psframe(1.65,1.65)(1.75,1.75)
\psline(0.2,0.2)(4,2)
\psline(0.2,0.2)(2,4)
\pscurve(4,2)(1.7,1.7)(2,4)
\uput[l](0,1.7){$U_2(x^*,y^*)$}
\uput[l](0,2){$U_2(S_1,NS_2)$}
\uput[l](0,4){$U_2(NS_1,S_2)$}
\uput[l](0,0.2){$U_2(S_1,S_2)$}
\uput[l](0,0.5){$U_2(NS_1,NS_2)$}
\uput[d](1.7,0){$U_1(x^*,y^*)$}
\uput[d](2.1,-0.3){$U_1(S_1,NS_2)$}
\uput[d](4,0){$U_1(NS_1,S_2)$}
\uput[d](0.2,0){$U_1(S_1,S_2)$}
\uput[d](0.7,-0.3){$U_1(NS_1,NS_2)$}
\uput[d](3,4.5){Correlated Equilibria}
\psline[linewidth=0.5pt]{->}(3,4.1)(3,3.1)
\psline[linestyle=dashed,linewidth=0.5pt](2,0)(2,4)
\psline[linestyle=dashed,linewidth=0.5pt](0,4)(2,4)
\psline[linestyle=dashed,linewidth=0.5pt](4,0)(4,2)
\psline[linestyle=dashed,linewidth=0.5pt](0,2)(4,2)
\psline[linestyle=dashed,linewidth=0.5pt](1.7,0)(1.7,1.7)
\psline[linestyle=dashed,linewidth=0.5pt](0,1.7)(1.7,1.7)
\psline[linestyle=dashed,linewidth=0.5pt](0.2,0)(0.2,0.2)
\psline[linestyle=dashed,linewidth=0.5pt](0,0.2)(0.2,0.2)
\psline[linestyle=dashed,linewidth=0.5pt](0.5,0)(0.5,0.5)
\psline[linestyle=dashed,linewidth=0.5pt](0,0.5)(0.5,0.5)
\psline[linewidth=2.7pt](2,4)(4,2)
\end{pspicture}
\end{center}
\caption{The Correlated Equilibria.}
\end{figure}
\end{footnotesize}

\section{Conclusion}

In this paper we have introduced a new power control game where the
action of a player is hybrid, one component is discrete while the
other is continuous. Whereas the general study of these games
remains to be done, it turns out that in our case we can prove the
existence of a Braess paradox which allows us to restrict our
attention to two separate games played consecutively: a finite game
where the players decide to sense or not and a compact game where
the transmitter chooses his power level. We have studied in details
the sensing game. In particular, it is proved it is weighted
potential. Also, by characterizing the correlated equilibria of this
game we show what is achievable in terms of fairness. Much work
remains to be done to generalize all these results to games with
arbitrary number of players and conduct simulations in relevant
wireless scenarios.

\section{Appendix 1}

\begin{eqnarray*}
&&\alpha<\frac{\beta^* - \gamma^*}{1-\beta^*\gamma^*}\\
&\Longleftrightarrow & \frac{1-\gamma^*\beta^* - \beta^*-\gamma^* }{(1-\gamma^*\beta^*)} <  1- \alpha \\
&\Longleftrightarrow & (1-\beta^*)( 1 +\gamma^*) < (1 - \alpha)[(1-\beta^*)(1+\gamma^*) + \gamma^* + \beta^* ] \\
&\Longleftrightarrow & \frac{f(\beta^*)}{\beta^*}(1-\beta^*) < (1 - \alpha)\frac{f(\beta^*)}{\beta^*}\frac{ 1 - \beta^*\gamma^*}{ 1 +\gamma^*} \\
&\Longleftrightarrow & \frac{R_1g_1f(\beta^*)(1-\beta^*)}{\sigma^2\beta^*} < (1-\alpha)\frac{R_1g_1f(\beta^*)(1-\gamma^*\beta^*)}{\sigma^2\beta^*(1+\gamma^*)}
\end{eqnarray*}

\section{Appendix 2}
\begin{figure*}[ht]
\begin{eqnarray*}
&&\frac{R_1g_1f(\beta^*)(1-\beta^*)}{\sigma^2\beta^*} \cdot y^* +  \frac{R_1g_1f(\gamma^*)(1-\gamma^*\beta^*)}{\sigma^2\gamma^*(1+\beta^*)} \cdot (1- y^*) \\
&=&  (1-\alpha)\frac{R_1g_1f(\beta^*)(1-\gamma^*\beta^*)}{\sigma^2\beta^*(1+\gamma^*)}  \cdot y^* + (1-\alpha)\frac{R_1g_1f(\beta^*)(1-\beta^*)}{\sigma^2\beta^*}  \cdot (1- y^*) \\
\Longleftrightarrow  &&  y^*  \cdot [ \frac{R_1g_1f(\beta^*)(1-\beta^*)}{\sigma^2\beta^*}  -  (1-\alpha)\frac{R_1g_1f(\beta^*)(1-\gamma^*\beta^*)}{\sigma^2\beta^*(1+\gamma^*)}\\
&&  + (1-\alpha)\frac{R_1g_1f(\beta^*)(1-\beta^*)}{\sigma^2\beta^*} -  \frac{R_1g_1f(\gamma^*)(1-\gamma^*\beta^*)}{\sigma^2\gamma^*(1+\beta^*)} ] \\
&=&
(1-\alpha)\frac{R_1g_1f(\beta^*)(1-\beta^*)}{\sigma^2\beta^*}
- \frac{R_1g_1f(\gamma^*)(1-\gamma^*\beta^*)}{\sigma^2\gamma^*(1+\beta^*)} \\
\Longleftrightarrow && y^* = \frac{(1-\alpha)\frac{f(\beta^*)}{\beta^*}(1-\beta^*) - \frac{f(\gamma^*)}{\gamma^*} \frac{1- \gamma^*\beta^*}{ 1+ \beta^*}  }{(1-\alpha)\frac{f(\beta^*)}{\beta^*}(1-\beta^*) - \frac{f(\gamma^*)}{\gamma^*} \frac{1- \gamma^*\beta^*}{ 1+ \beta^*}   + \frac{f(\beta^*)}{\beta^*}(1-\beta^*) - (1-\alpha) \frac{f(\beta^*)}{\beta^*} \frac{1- \gamma^*\beta^*}{ 1+ \gamma^*} }
\end{eqnarray*}
\end{figure*}

Replacing the above $y^*$ into the indifference equation, we obtain the utility of player 1 at the mixed equilibrium.
\begin{figure*}[ht]
\begin{eqnarray*}
U_1(x^*,y^*) &=& \frac{\frac{R_1g_1f(\beta^*)(1-\beta^*)}{\sigma^2\beta^*}\frac{R_1g_1f(\gamma^*)(1-\gamma^*\beta^*)}{\sigma^2\gamma^*(1+\beta^*)}
- \frac{R_1g_1}{\sigma^2}\frac{f(\gamma^*)}{\gamma^*} \frac{1- \gamma^*\beta^*}{ 1+ \beta^*}\frac{R_1g_1}{\sigma^2} (1-\alpha) \frac{f(\beta^*)}{\beta^*} \frac{1- \gamma^*\beta^*}{ 1+ \gamma^*}
}{ \frac{R_1g_1}{\sigma^2}(1-\alpha)\frac{f(\beta^*)}{\beta^*}(1-\beta^*) -  \frac{R_1g_1}{\sigma^2}\frac{f(\gamma^*)}{\gamma^*} \frac{1- \gamma^*\beta^*}{ 1+ \beta^*}
+ \frac{R_1g_1}{\sigma^2}\frac{f(\beta^*)}{\beta^*}(1-\beta^*) -  \frac{R_1g_1}{\sigma^2}(1-\alpha) \frac{f(\beta^*)}{\beta^*} \frac{1- \gamma^*\beta^*}{ 1+ \gamma^*} }\\
&+& \frac{ \frac{R_1g_1}{\sigma^2}(1-\alpha)\frac{f(\beta^*)}{\beta^*}(1-\beta^*)\frac{R_1g_1}{\sigma^2}\frac{f(\beta^*)}{\beta^*}(1-\beta^*) -\frac{R_1g_1f(\beta^*)(1-\beta^*)}{\sigma^2\beta^*}\frac{R_1g_1f(\gamma^*)(1-\gamma^*\beta^*)}{\sigma^2\gamma^*(1+\beta^*)}
}{ \frac{R_1g_1}{\sigma^2}(1-\alpha)\frac{f(\beta^*)}{\beta^*}(1-\beta^*) -  \frac{R_1g_1}{\sigma^2}\frac{f(\gamma^*)}{\gamma^*} \frac{1- \gamma^*\beta^*}{ 1+ \beta^*}
+ \frac{R_1g_1}{\sigma^2}\frac{f(\beta^*)}{\beta^*}(1-\beta^*) -  \frac{R_1g_1}{\sigma^2}(1-\alpha) \frac{f(\beta^*)}{\beta^*} \frac{1- \gamma^*\beta^*}{ 1+ \gamma^*} }\\
&=& \frac{R_1g_1}{\sigma^2} \frac{(1-\alpha)\frac{f(\beta^*)}{\beta^*}(1-\beta^*)\frac{f(\beta^*)}{\beta^*}(1-\beta^*) - \frac{f(\gamma^*)}{\gamma^*} \frac{1- \gamma^*\beta^*}{ 1+ \beta^*} (1-\alpha) \frac{f(\beta^*)}{\beta^*} \frac{1- \gamma^*\beta^*}{ 1+ \gamma^*}
}{(1-\alpha)\frac{f(\beta^*)}{\beta^*}(1-\beta^*) - \frac{f(\gamma^*)}{\gamma^*} \frac{1- \gamma^*\beta^*}{ 1+ \beta^*}   + \frac{f(\beta^*)}{\beta^*}(1-\beta^*) - (1-\alpha) \frac{f(\beta^*)}{\beta^*} \frac{1- \gamma^*\beta^*}{ 1+ \gamma^*} }
\end{eqnarray*}
\end{figure*}
The same argument applies:
\begin{figure*}[ht]
\begin{eqnarray*}
U_2(x^*,y^*) = \frac{R_2g_2}{\sigma^2} \frac{(1-\alpha)\frac{f(\beta^*)}{\beta^*}(1-\beta^*)\frac{f(\beta^*)}{\beta^*}(1-\beta^*) - \frac{f(\gamma^*)}{\gamma^*} \frac{1- \gamma^*\beta^*}{ 1+ \beta^*} (1-\alpha) \frac{f(\beta^*)}{\beta^*} \frac{1- \gamma^*\beta^*}{ 1+ \gamma^*}
}{(1-\alpha)\frac{f(\beta^*)}{\beta^*}(1-\beta^*) - \frac{f(\gamma^*)}{\gamma^*} \frac{1- \gamma^*\beta^*}{ 1+ \beta^*}   + \frac{f(\beta^*)}{\beta^*}(1-\beta^*) - (1-\alpha) \frac{f(\beta^*)}{\beta^*} \frac{1- \gamma^*\beta^*}{ 1+ \gamma^*} }
\end{eqnarray*}
\end{figure*}


\section{Appendix 3}
\begin{eqnarray*}
&&\alpha < \frac{f(\beta^*) - f(\gamma^*) + \frac{f(\beta^*)}{\beta^*} - \frac{f(\gamma^*)}{\gamma^*}  }{f(\beta^*)\frac{1+\beta^*}{\beta^*}}\\
&\Longleftrightarrow&  1  -  \alpha > \frac{ f(\beta^*)\frac{1+\beta^*}{\beta^*}  -  f(\gamma^*)\frac{1+\gamma^*}{\gamma^*}   }{f(\beta^*)\frac{1+\beta^*}{\beta^*}}\\
&\Longleftrightarrow&  (1  -  \alpha)  \frac{f(\beta^*)}{\beta^*} \frac{1- \gamma^*\beta^*}{ 1+ \beta^*}  > \frac{f(\gamma^*)}{\gamma^*} \frac{1- \gamma^*\beta^*}{ 1+ \gamma^*}
\end{eqnarray*}

\section{Appendix 4}
The proof comes from the theorem of Monderer and Shapley 1996 (see Sandholm "Decomposition of Potential" 2010)
\begin{theorem}
The game $G$ is a potential game if and only if for every players $i,j\in K$, every pair of actions $s_i,t_i\in S_i$ and $s_j,t_j\in S_j$ and every joint action $s_k\in S_{K\backslash \{i,j\}}$, we have that
\begin{eqnarray*}
&& U_i(t_i,s_j,s_k)- U_i(s_i,s_j,s_k) + U_i(s_i,t_j,s_k) -U_i(t_i,t_j,s_k) +\\
&& U_j(t_i,t_j,s_k)- U_j(t_i,s_j,s_k) +U_j(s_i,s_j,s_k) -U_j(s_i,t_j,s_k) =0
\end{eqnarray*}
\end{theorem}

Let us prove that the  two conditions provided by our theorem are equivalent to the one of Monderer and Shapley's theorem.
We introduce the following notation defined for each player $i\in K$ and each action $T\in \mathcal{S}$.
\begin{eqnarray}
w_i&=&R_ig_i\\
U^T(t_i,t_j,s_k)&=&\frac{U_i^T(t_i,t_j,s_k)}{w_i}
\end{eqnarray}
For every players $i,j\in K$, every pair of actions $s_i,t_i\in S_i$ and $s_j,t_j\in S_j$ and every joint action $s_k\in S_{K\backslash \{i,j\}}$, we have the following equivalences:
\begin{eqnarray*}
&&U_i(t_i,s_j,s_k)- U_i(s_i,s_j,s_k)\\
&& + U_i(s_i,t_j,s_k) -U_i(t_i,t_j,s_k) \\
&+ & U_j(t_i,t_j,s_k)- U_j(t_i,s_j,s_k)\\
&& +U_j(s_i,s_j,s_k) -U_j(s_i,t_j,s_k)=0\\
\Longleftrightarrow
&&w_i(U^{T}(t_i,s_j,s_k)- U^S(s_i,s_j,s_k)\\
&& + U^S(s_i,t_j,s_k) -U^T(t_i,t_j,s_k)) \\
&+& w_j(U^T(t_i,t_j,s_k)- U^S(t_i,s_j,s_k)\\
&& +U^S(s_i,s_j,s_k) -U^T(s_i,t_j,s_k))=0\\
\Longleftrightarrow
&&(w_i-w_j)(U^{T}(t_i,s_j,s_k)- U^S(s_i,s_j,s_k)\\
&& + U^S(s_i,t_j,s_k) -U^T(t_i,t_j,s_k))=0 \\
&&
\Longleftrightarrow
\begin{cases}
w_i = w_j \\
U^{T}(t_i,s_j,s_k)- U^S(s_i,s_j,s_k)\\
 + U^S(s_i,t_j,s_k) -U^T(t_i,t_j,s_k)= 0 \\
\end{cases}
\end{eqnarray*}
Thus the sensing game is a potential game if and only if one of the two following condition is satisfied:
\begin{eqnarray}
&&\forall i,j\in K\quad R_ig_i=R_jg_j\\
&&\forall i,j\in K,\; s_i,t_i\in S_i,\;\forall s_j,t_j\in S_j,\;\forall s_k\in S_{K\backslash \{i,j\}}\\
&&U^{T}(t_i,s_j,s_k)- U^S(s_i,s_j,s_k) \\
&&+ U^S(s_i,t_j,s_k) -U^T(t_i,t_j,s_k)= 0
\end{eqnarray}

\section{Appendix 5}
The proof of this theorem follows the same line of the previous theorem. It suffices to show that the auxiliary game defined as follows is a potential game.
\begin{eqnarray}
\widetilde{G}=(K,(\mathcal{S})_{i\in K},(\tilde{U}_i)_{i\in K})
\end{eqnarray}
Where the utility are defined by the following equations with $w_i = \frac{R_ig_i}{\sigma^2}$.
\begin{eqnarray}
\tilde{U}_i(s_i,s_{-i})&=&\frac{U_i(s_i,s_{-i})}{w_i}
\end{eqnarray}
From the above demonstration, it is easy to show that, for every players $i,j\in K$, every pair of actions $s_i,t_i\in S_i$ and $s_j,t_j\in S_j$ and every joint action $s_k\in S_{K\backslash \{i,j\}}$:
\begin{eqnarray}
&&\tilde{U}_i(t_i,s_j,s_k)- \tilde{U}_i(s_i,s_j,s_k)\\
& +& \tilde{U}_i(s_i,t_j,s_k) -\tilde{U}_i(t_i,t_j,s_k) \\
&+ & \tilde{U}_j(t_i,t_j,s_k)- \tilde{U}_j(t_i,s_j,s_k)\\
&+& \tilde{U}_j(s_i,s_j,s_k) -\tilde{U}_j(s_i,t_j,s_k)=0
\end{eqnarray}
We conclude that the sensing game is a weighted potential game.

\bibliographystyle{plain}
\bibliography{BiblioMael}

\end{document}